\documentclass[letterpaper,11pt]{scrartcl}

\usepackage[margin=1in,bottom=1.2in]{geometry}
\usepackage{amsfonts,amssymb,amsmath}
\usepackage[thmmarks,hyperref,amsthm,amsmath]{ntheorem} 
\usepackage{newpxtext,newpxmath}
\usepackage{graphicx}
\usepackage[usenames,dvipsnames]{color}
\usepackage{hyperref}
\usepackage[utf8]{inputenc}
\usepackage{authblk}
\usepackage{xspace}
\usepackage{mathtools}
\usepackage{marvosym}
\usepackage{prettyref}
\usepackage{microtype}

\definecolor{blueLink}{rgb}{0,0.2,0.8}
\hypersetup{colorlinks,linkcolor=blueLink,urlcolor=blueLink,citecolor=blueLink}

\DeclareUrlCommand\email{\urlstyle{rm}}

\let\pref=\prettyref
\newtheorem{theorem}{Theorem}
\newtheorem{lemma}[theorem]{Lemma}
\newtheorem{corollary}[theorem]{Corollary}
\newrefformat{cor}{Corollary~\ref{#1}}
\newrefformat{prop}{Property~\ref{#1}}

\newcommand{\N}{\mathbb N}
\newcommand{\cR}{\mathcal{R}}
\newcommand{\cD}{\mathcal{D}}
\newcommand{\ADV}{\textsc{Adv}\xspace}
\newcommand{\ALG}{\textsc{Alg}\xspace}
\newcommand{\OPT}{\textsc{Opt}\xspace}
\newcommand{\X}{\mathcal{X}}
\newcommand{\T}{\mathcal{T}}

\newcommand{\oldd}{b}
\newcommand{\ih}{i_h}

\title{Unbounded lower bound for k-server \\ against weak adversaries\thanks{%
Supported by Polish National Science Centre grants 
2015/18/E/ST6/00456, 2016/22/E/ST6/00499, 2016/21/D/ST6/02402,
and the NWO VICI grant 639.023.812.}}

\author[1]{Marcin Bienkowski}
\author[1]{Jaros{\l}aw Byrka}
\author[2]{Christian Coester}
\author[1]{{\L}ukasz Je\.{z}}
\affil[1]{Institute of Computer Science, University of Wrocław, Poland\protect\\
\{marcin.bienkowski,jaroslaw.byrka,lukasz.jez\}\MVAt{}cs.uni.wroc.pl}
\affil[2]{CWI, Amsterdam, Netherlands\protect\\
christian.coester\MVAt{}cwi.nl}
\date{}

\begin{document}

\maketitle

\begin{abstract}
	We study the resource augmented version of the $k$-server problem, also known as
	the $k$-server problem against weak adversaries or the $(h,k)$-server
	problem. In this setting, an online algorithm using $k$ servers is compared
	to an offline algorithm using $h$ servers, where $h\le k$. For uniform
	metrics, it has been known since the seminal work of Sleator and Tarjan
	(1985) that for any $\epsilon>0$, the competitive ratio drops to a constant
	if $k=(1+\epsilon) \cdot h$. This result was later generalized to weighted stars
	(Young 1994) and trees of bounded depth (Bansal et al. 2017). The main open
	problem for this setting is whether a similar phenomenon occurs on general
	metrics.
	
	We resolve this question negatively. With a simple recursive construction,
	we show that the competitive ratio is at least $\Omega(\log \log h)$, even
	as $k\to\infty$. Our lower bound holds for both deterministic and randomized
	algorithms. It also disproves the existence of a~competitive algorithm for
	the infinite server problem on general metrics.
\end{abstract}

%%%%%%%%%%%%%%%%%%%%%%%%%%%%%%%%%%%%%%%%%%%%%%%%%%%%%%%%%%%%%%%%%%%%%%%%%%%%%%%%%%
%%%%%%%%%%%%%%%%%%%%%%%%%%%%%%%%%%%%%%%%%%%%%%%%%%%%%%%%%%%%%%%%%%%%%%%%%%%%%%%%%%

\section{Introduction}

The $k$-server problem is one of the most well-studied and influential online
problems in competitive analysis, defined in 1990 by Manasse et
al.~\cite{MaMcSl90}. It generalizes many problems in which an~algorithm has to
maintain a feasible state while satisfying a sequence of requests. Formally, the
$k$-server problem is defined as follows. There are $k$ servers in a metric
space $(\X, d)$ and a~sequence $r_1, r_2, r_3, \ldots$ of requests to metric
space points appears online. In response to a~request~$r_i$, an~algorithm has to
move its servers, so that one of them ends at point $r_i$. The goal is to
minimize the cost defined as the total distance traveled by all servers.

\subsection{From Uniform to General Metrics}

The definition of the $k$-server problem is deceivingly simple, but it has led
to substantial progress in many branches of competitive analysis. Historically,
the results were obtained first for the case where $\X$ is a~uniform metric
space: the $k$-server problem is then equivalent to the paging problem with a
cache of size~$k$~\cite{SleTar85}. In particular, the competitive ratio for
paging is $k$ for deterministic algorithms and there is a lower bound of $k$
that holds for arbitrary metric spaces of more than $k$ points~\cite{MaMcSl90}.
This led to the bold \emph{$k$-server conjecture}~\cite{MaMcSl90} stating that
this ratio is $k$ for all metric spaces. After series of papers proving the
upper bound of $k$ for particular metrics (e.g., trees or lines), the conjecture
has been positively resolved (in the asymptotic sense) by the celebrated $2k-1$
upper bound due to Koutsoupias and Papadimitriou~\cite{KouPap95}. For a more
thorough treatment of the history of deterministic approaches, see a survey by
Koutsoupias~\cite{Koutso09}.

Similarly, randomized competitive solutions for uniform
metrics~\cite{McGSle91,AcChNo00,FKLMSY91} showed that the achievable competitive
ratio is exactly $H_k = \Theta(\log k)$ and led to the analogous
\emph{randomized $k$-server conjecture}, stating that the randomized competitive
ratio is $\Theta(\log k)$ on arbitrary metrics. Some cornerstone results towards
resolving this conjecture deserve closer attention. On the lower bound side,
Bartal et al.~\cite{BaBoMe06} used Ramsey-type phenomena for metric spaces to
show that the randomized competitive ratio is at least $\Omega(\log k / \log
\log k)$ for any metric space.\footnote{In the description of all lower bounds
on the competitive ratio for the $k$-server problem, we silently assume that the
metric space in question has more than $k$ points.} On the algorithmic side, 
a~major breakthrough (building on a long line of results for particular metrics)
was obtained by Bansal et al.~\cite{BaBuMN15}, who constructed an algorithm of
ratio poly-logarithmic in the number of metric space points, based on HST
embeddings (hierarchically separated trees) and the so-called fractional
allocation problem. It was recently improved by Bubeck et al.~\cite{BuCLLM18},
who used mirror descent dynamics with multi-scale entropic regularization to
obtain an~$O(\log^2 k)$-competitive algorithm on HSTs and an $O(\log^2k\log
n)$-competitive algorithm on general $n$-point metrics. Based on this,
Lee~\cite{Lee18} proposed a dynamic embedding technique to achieve a competitive
ratio poly-logarithmic in $k$ on arbitrary metrics.

\subsection{Weak Adversaries}

A way to compensate for the online algorithm's lack of knowledge of the future
is to assume that the  algorithm has more ``resources'' than the offline optimum
it is compared to. This natural concept, called \emph{resource augmentation},
has led to spectacular successes for online scheduling problems (see
e.g.~\cite{KalPru00,PhStTW02}). It can be a way to overcome pessimistic
worst-case bounds of the original setting. In the context of the $k$-server
problem, it is also known as the \emph{weak adversaries}
model~\cite{Koutso99,BaEJKP18} or the $(h,k)$-server problem: an~online
algorithm with $k$ servers is compared to an  optimal algorithm (an~adversary)
with $h \leq k$ servers.  For a metric space $\X$, let $\cD_\X(h,k)$ and
$\cR_\X(h,k)$ denote the best competitive ratios of deterministic and randomized
algorithms, respectively, for the $(h,k)$-server problem on $\X$.

Again, the first results for the $(h,k)$-server problem were developed for
uniform metrics: Sleator and Tarjan~\cite{SleTar85} gave an exact answer of
$\cD_\X(h,k)=k/(k-h+1)$, with the upper bound being achieved by the LRU (least
recently used) paging strategy. This implies that having $k = (1+\epsilon) \cdot
h$ servers suffices to attain a constant competitive ratio. It is natural to ask
whether such phenomenon extends to other metrics. This question was raised
already by Manasse et al~\cite{MaMcSl90} when they introduced the $k$-server
problem.

Formally, we study the following questions.
\begin{description}
\item[Strong $(h,k)$-server hypothesis:] \emph{For any metric space~$\X$ and any
	$\epsilon>0$, $\cD_\X(h,k)=O_\epsilon(1)$ whenever $k\ge (1+\epsilon) \cdot h$.}
\item[Weak $(h,k)$-server hypothesis:] \emph{For any metric space $\X$ and any
	$h\in\N$, $\cD_\X(h,k)=O(1)$ as~$k\to\infty$.}
\end{description}

Generalizing the result for uniform metrics, the same competitive ratio of
$k/(k-h+1)$ was later also obtained for weighted star metrics~\cite{Young94}.
More recently, Bansal et al.~\cite{BaElJK19} confirmed the strong $(h,k)$-server
hypothesis also for trees of bounded depth. Using
randomization, tight bounds of $\cR_\X(h,k)=\Theta(\log(1/\epsilon))$ were
obtained for uniform metrics~\cite{Young91} and weighted stars~\cite{BaBuNa12a}
when $k=(1+\epsilon) \cdot h$. The recent results by Bubeck et al.~\cite{BuCLLM18} and
Buchbinder et al.~\cite{BuGuMN19} for the $k$-server problem extend also to the
$(h,k)$-server setting, implying that $\cR_\X(h,k) = O(D \cdot
\log(1/\epsilon))$ for HSTs of depth $D$ when $k=(1+\epsilon) \cdot h$.\footnote{For
general trees of depth $D$, they obtain a \emph{fractional} algorithm achieving
the same competitive ratio.}

Surprisingly, the performance of some classical algorithms can slightly degrade
when additional online servers are available. Bansal et al.
\cite{BaEJKP18,BaElJK19} showed that this can occur for both the Work Function
algorithm and the Double Coverage algorithm. On the positive side,
Koutsoupias~\cite{Koutso99} showed that the Work Function algorithm obtains a
competitive ratio of at most $2h$ simultaneously for all $h\le k$. The algorithm
of \cite{BaElJK19} confirming the $(h,k)$-server hypothesis on bounded depth
trees is actually a variant of the Double Coverage algorithm.

In \cite{CoKoLa17}, the infinite server problem (denoted $\infty$-server problem
here) has been introduced as a possible way to resolve the question on general
metrics. This is the variant of the $k$-server problem where $k=\infty$, and all
infinitely many servers initially reside at the same point of the metric space.
The existence of an $O(1)$-competitive algorithm for the $\infty$-server problem
was shown to be equivalent to an affirmative resolution of the weak
$(h,k)$-server hypothesis.

In terms of lower bounds, it is known that --- unlike in the case of uniform and
weighted star metrics --- the ratio $\cD_\X(h,k)$ does \emph{not} converge to
$1$ on general metrics even as $k\to\infty$. Namely, Bar-Noy and
Schieber~\cite[page 175]{BorEl-98} showed that $\cD_\X(2,k) = 2$ for all $k$
when $\X$ is the line metric. For large $h$, the lower bound on $\cD_\X(h,k)$
was improved to $2.41$ \cite{BaElJK19} using depth-2 trees and later to
$3.14$~\cite{CoKoLa17} by a reduction from the $\infty$-server problem. In the
absence of any super-constant lower bounds, the $(h,k)$-server hypothesis
continued to seem plausible. In fact, Bansal et al.~\cite{BaElJK19} argued that
it would be very surprising if $\cD_\X(h,k) = \omega(1)$ (even for a
sufficiently large~$k$).

\subsection{Our Results}

Our main result is the refutation of both versions of the $(h,k)$-server hypothesis:

\begin{theorem}\label{thm:main}
There exists a tree metric $\T$ such that $\cR_\T(h,k) = \Omega(\log\log h)$, 
even for arbitrarily large~$k$.
\end{theorem}

Since $\cD_\T(h,k)\ge \cR_\T(h,k)$, the lower bound obviously extends to
deterministic algorithms. The underlying construction is simple. It is based on
recursively combining Young's lower bound for randomized
$(h,k)$-paging~\cite{Young91} along many scales. At higher scales, the
construction is applied to groups of servers rather than individual servers.

Due to the connection between the $(h,k)$-server problem and the $\infty$-server
problem~\cite{CoKoLa17}, a~direct consequence of Theorem~\ref{thm:main} is that
there is no competitive algorithm for the $\infty$-server problem on general
metrics. In fact, we first found the lower bound by analyzing the
$\infty$-server problem.

\begin{corollary}\label{cor:infty}
	The competitive ratio of the $\infty$-server problem on trees of depth $D$ is $\Omega(\log D)$. In particular, no algorithm for the $\infty$-server problem on general metrics has a constant competitive ratio.
\end{corollary}

\subsection{Preliminaries}

An online algorithm $\ALG$ is called $\rho$-competitive if
\begin{align*}
\ALG(\sigma) \le \rho \cdot \OPT(\sigma)+C 
\end{align*}
for all request sequences $\sigma$, where $\ALG(\sigma)$ and $\OPT(\sigma)$
denote the cost of $\ALG$ and the optimal (offline) cost for $\sigma$,
respectively, and ${C\ge 0}$ is a constant independent of $\sigma$. The
competitive ratio of a~problem is the infimum of all $\rho$ such that a
$\rho$-competitive algorithm exists. In the case of randomized algorithms,
$\ALG(\sigma)$ is replaced by its expectation. Note that for the 
$(h,k)$-server problem, \OPT denotes the optimal solution using $h$ servers, while \ALG uses $k$ servers.

An algorithm is \emph{fractional} if it is allowed to move an arbitrary fraction of a server,
paying the same fractions of the distance moved, but it is still required to
bring ``a total mass'' of at least one server to the requested point. A
fractional algorithm can be derived from a randomized one by setting the server
mass at each point to the expected number of servers; clearly, the cost of the fractional algorithm is at most the expected cost of the randomized algorithm.\footnote{On weighted stars and HST metrics, the converse is also true: Any fractional algorithm can be rounded online to a randomized integral one while increasing its cost by at most a multiplicative constant~\cite{BaBuNa12a,BaBuMN15}. It is unknown whether this also holds for general metrics.}

All metric spaces constructed in this paper are trees with a distinguished root, and we assume that servers reside initially at the root. We will charge cost (to both the online and offline algorithms) only for traversing edges in direction \emph{away} from the root. Since movement away from the root is within a factor $2$ of the total movement, the error due to this is absorbed in the asymptotic notation of our results.

For an infinite request sequence $\sigma$, we denote its prefix of the first $m$
requests by $\sigma_m$. 

%%%%%%%%%%%%%%%%%%%%%%%%%%%%%%%%%%%%%%%%%%%%%%%%%%%%%%%%%%%%%%%%%%%%%%%%%%%%%%%%%%
%%%%%%%%%%%%%%%%%%%%%%%%%%%%%%%%%%%%%%%%%%%%%%%%%%%%%%%%%%%%%%%%%%%%%%%%%%%%%%%%%%

\section{Proof of the lower bound}

Below we state the main lemma towards proving 
\pref{thm:main}.

\begin{lemma}\label{lem:main}
Fix arbitrary $\rho\ge 1$, $\delta>0$ and an integer $i \geq 0$. Let $\oldd=
\lceil \exp(3\rho) \rceil$, $h_i=\oldd^i$, $k_i= \oldd^{i}\cdot
\left(1+i/(2\oldd)\right)$. There exists a tree $T_i$ of depth $i$ such that,
for any fractional online $k_i$-server algorithm $\ALG$, there exists an
infinite request sequence $\sigma$ on $T_i$ satisfying two properties:
\begin{enumerate} %Lukasz: do not use the labels below!  Without an extra package, they refer to lemma number!
\item[(a)] 
$\ALG(\sigma_m)\ge \rho \cdot\OPT_{h_i}(\sigma_m)-\delta$ for all $m\in\N$, where $\OPT_{h_i}(\sigma_m)$ denotes the optimal cost for serving $\sigma_m$ using $h_i$ servers.
%\label{prop:a}
\item[(b)] 
 If $i\ge 1$, then $\ALG(\sigma_m)\to\infty$ as $m\to\infty$.
%\label{prop:b}
\end{enumerate}
\end{lemma}

\begin{proof}
We prove the lemma by induction on $i$. For $i=0$, the tree $T_0$ is simply a single node, and all requests are given at this node. Clearly, the lemma holds here.

For the inductive step, we fix any $i \geq 0$. We will show that the lemma properties for $i+1$ hold 
for \emph{some} $\delta$. By 
scaling all distances by a small multiplicative constant, this
implies that $\delta$ can be made arbitrarily close to $0$,
yielding the lemma statement for $i+1$ and arbitrary $\delta$. 
	
Let $T_i$ be the tree induced by the induction hypothesis for $\delta = 1/2$. 
The root of $T_{i+1}$ has infinitely many children at distance $1$; all the
subtrees rooted at these children are copies of $T_i$. We will assume that the server mass in each subtree never exceeds $k_i$; this assumption will be justified later. It allows us to invoke the induction
hypothesis on the subtrees. If the mass inside a subtree is $k_i-c$ for some $c\ge
0$, we interpret this as mass~$c$ sitting at the root of the
subtrees. Note that the sub-algorithms for the different
subtrees are not independent of each other, as a request in one subtree can
trigger movement towards the root in another subtree. However, we construct the
request sequence in an~online manner where each request is independent of
decisions of the algorithm for future requests, and thus we can analyze the
sub-algorithms independently of each other.

Let $\epsilon>0$ be some small constant. The request sequence $\sigma$ consists
of phases numbered from $1$. In each phase, $\oldd$ subtrees $T_i$ will be marked
and among them $\oldd-1$ subtrees were marked in the previous phase. 
For this definition, we assume that right before phase $1$, in an artificial phase $0$ containing no requests, 
$\oldd$ arbitrary subtrees were marked. 
All phases proceed as follows:
\begin{itemize}
	\item Mark a fresh subtree $T_i$ that has never received any requests before.
	\item While the server mass in the fresh subtree is at most $k_i-\epsilon$, issue requests in it according to the induction hypothesis.
	\item For $j=1,\dots,\oldd-1$:
	\begin{itemize}
		\item Among the subtrees that were marked in the last phase but have not been marked (yet) in the current phase, mark the one with the least server mass.
		\item While there exists a subtree marked in the current phase where the server mass 
		is at most $k_i-\epsilon$, issue requests in this subtree according to the induction hypothesis.
	\end{itemize}
\end{itemize}

The request sequence satisfies Property (b): if $i=0$, then each request incurs at least cost $\epsilon$, and if $i\ge 1$ this follows by the induction hypothesis.
We now prove that Property (a) also holds.

We compare $\ALG$ against an offline algorithm \ADV with $h_{i+1} = \oldd^{i+1}$ servers that always has 
$h_i = \oldd^{i}$ servers at each marked subtree of the current phase, and uses servers optimally within the subtrees.

Consider some phase. Denote by $\ALG_{\ell}$ and $\ALG_{\le \ell}$ the cost of
$\ALG$ incurred during the phase along edges of level $\ell$ and along edges of level at most~$\ell$,
respectively. We define $\ADV_\ell$ and $\ADV_{\le \ell}$ analogously.
Here, we use the convention that edges incident to the leaves have
level $1$ and edges incident to the root have level $i+1$. 

Consider the case that the phase under consideration is complete. We analyze
first the cost along edges incident to the root. $\ALG$ pays at least $k_i-\epsilon$ to
move server mass $k_i-\epsilon$ to the fresh subtree. At the beginning of iteration $j$
of the for-loop, $\ALG$ has server mass at least $k_i-\epsilon$ in each of the 
$j$~subtrees that were marked during the current phase. Note that $j-1$ of them 
were marked in the previous phase.  Thus, the average amount of
server mass in the $\oldd-(j-1)$ subtrees that were marked in the last phase but
not yet in the current phase is at most $(k_{i+1}-j \cdot (k_i-\epsilon)) / (\oldd-j+1)$.
In effect, the cost to move mass to the subtree that is marked in the $j$th iteration is at least
\begin{align*}
k_i-\epsilon-\frac{k_{i+1}-j \cdot (k_i-\epsilon)}{\oldd-j+1}=\frac{(k_i-\epsilon)(\oldd+1)-k_{i+1}}{\oldd-j+1}.
\end{align*}
Hence, as $\epsilon\to 0$, the total cost of moving server mass to the marked subtrees of the phase is at least
\begin{align*}
\ALG_{i+1}&\ge k_i-o(1)+\sum_{j=1}^{\oldd-1}\frac{k_i(\oldd+1)-k_{i+1}}{\oldd-j+1} \\
&= \oldd^{i} \cdot \left(1+\frac{i}{2\oldd}+\sum_{j=1}^{\oldd-1}\frac{\left(1+\frac{i}{2\oldd}\right)(\oldd+1)-\oldd\left(1+\frac{i+1}{2\oldd}\right)}{\oldd-j+1}\right)-o(1)\\
&= \oldd^{i} \cdot \left(\frac{1}{2} + \left(\frac{1}{2}+\frac{i}{2\oldd} \right)
	+\sum_{j=1}^{\oldd-1}\frac{\frac{1}{2}+\frac{i}{2\oldd}}{\oldd-j+1}\right)-o(1)\\
&\ge \oldd^{i} \cdot \left(\frac{1}{2} + \frac{\ln \oldd}{3}\right)\\
&\ge \oldd^i\rho+ \oldd^i/2.
\end{align*}
In contrast, the offline cost during the phase along edges incident to the root is only
\begin{align*}
\ADV_{i+1}
&=\oldd^i
\end{align*}
because the offline algorithm moves only $\oldd^i$ servers from the subtree that was marked in the last but not in the current phase to the fresh subtree of the current phase.

For the cost within the subtrees, the induction hypothesis of (a) (applied to $\oldd$ marked subtrees with $\delta=1/2$) yields
\begin{align*}
\ALG_{\le i}\ge \rho\cdot\ADV_{\le i} - \oldd / 2.
\end{align*}
Therefore, for the total cost during a complete phase, we obtain
\begin{align*}
\ALG_{\le i+1} &= \ALG_{i+1} + \ALG_{\le i}\\
&\ge \rho\cdot\ADV_{i+1} + \oldd^i / 2 + \rho\cdot\ADV_{\le i} - \oldd / 2\\
&\ge \rho\cdot\ADV_{\le i+1}.
\intertext{In the last phase, which may be incomplete, we have}
\ALG_{\le i+1}&\ge \ALG_{\le i}\\
&\ge \rho\cdot\ADV_{\le i} - \oldd / 2\\
&\ge \rho\cdot\ADV_{\le i+1} - \rho \oldd^i - \oldd / 2.
\end{align*}
Now if we set $\delta'$ to be $\rho \oldd^i + \oldd / 2$ plus 
the cost of $\ADV$ to bring servers to the marked subtrees of phase $0$, then 
we obtain Property (a) for $i+1$ and a fixed $\delta'$. Recall that 
this yields the same property for arbitrary $\delta'$ by scaling 
all distances by a small factor.

Finally, it remains to justify the assumption that the server mass in each subtree $T_i$ never exceeds~$k_i$, which was necessary to allow invoking the induction hypothesis. Suppose after serving a request to a leaf~$u$, the algorithm ends up with server mass $k_i+c$ in the subtree that contains $u$, for some $c>0$. Call this subtree $S$. Without loss of generality, the distance from the root of $S$ to any leaf in $S$ is at most~$1$ (we can scale all subtrees $T_i$ down to achieve this). Upon serving the request at leaf $u$, at least mass $c+\epsilon$ traveled to leaf $u$ across the root $r$ of $T_{i+1}$. Consider an alternative algorithm $\ALG'$ that stores an amount $c$ of this mass at $r$ and brings mass $c$ from another vertex in $S$ to leaf $u$ instead. Upon serving this request, $\ALG'$ saves cost $c$ compared to $\ALG$, since it does not need to bring this mass from $r$ to the root of $S$. The next time that $\ALG$ would take mass from $S$ to another subtree $T_i$, we first take it from the mass $c$ that is stored at the root of $T_{i+1}$. Notice that there will be no request in $S$ until after the stored mass $c$ at $r$ has been used up. Once it has been used up, $\ALG'$ reorganizes its server mass so that its configuration is the same as that of $\ALG$ again. This requires at most cost $c$ for moving this much mass within $S$, which is the cost that $\ALG'$ had saved before. Thus, we have a new algorithm whose cost is less than that of $\ALG$, which never exceeds mass $k_i$ in any subtree (by repeating this idea), and for which our lower bound holds.
\end{proof}

We obtain the main result by combining the trees guaranteed by \pref{lem:main}.

\begin{proof}[Proof of \pref{thm:main}]
	For $i\in\N$, let $\oldd_i=\lfloor\sqrt{i}\rfloor$ and
	$\rho_i=\frac{1}{3}\ln \oldd_i$. The lower bound holds on the following tree
	$\T$: It contains as subtrees, for each $i\in\N$, infinitely many copies of
	the tree $T_i$ guaranteed by \pref{lem:main} for $\rho=\rho_i$ and $\delta=
	\oldd_i^i$. The roots of the subtrees $T_i$ are connected to the root of $\T$
	by edges of length $1$.
	
	Let $h\le k$ be the numbers of offline and online servers respectively. Let
	$\ih= \lfloor\sqrt{\ln h} \rfloor$. The adversarial sequence uses only
	subtrees of type $T_{\ih}$. It consists of epochs: 
	In each epoch, select a~subtree of type $T_{\ih}$ whose online
	server mass is zero.
	Requests are issued in this subtree as induced by \pref{lem:main}. As soon as the online server mass in the subtree exceeds $\oldd_{\ih}^{\ih}\cdot( 1+ \ih / (2 \oldd_{\ih}) )$, the epoch ends and a new epoch begins.
	
	At the start of each epoch, the offline algorithm brings $\oldd_{\ih}^{\ih}\le {\ih}^{\ih/2}\le \exp(\ih^2) \le h$ servers to the subtree of that epoch. For a given epoch, denote by $\ALG$ and $\ADV$, respectively, the online and offline cost suffered \emph{within} the active subtree of the epoch. By \pref{lem:main},
	\begin{align*}
	\ALG \ge \rho_{\ih} \cdot \ADV - \oldd_{\ih}^{\ih}.
	\end{align*}
	If the epoch runs indefinitely (because the algorithm never brings the required number of servers to the subtree), then the cost within the active subtree dominates the competitive ratio. Since $\rho_{\ih}=\Omega(\log \oldd_{\ih})=\Omega(\log \ih)=\Omega(\log\log h)$, the theorem follows.
	
	Otherwise, the online algorithm pays at least $\oldd_{\ih}^{\ih}\cdot( 1+ \ih / (2 \oldd_{\ih}) )$
	to bring as many servers to the subtree, whereas the offline algorithm pays only $\oldd_{\ih}^{\ih}$ to move servers to the subtree. Thus, the ratio of the total online to offline cost during each epoch is at least
	\begin{align*}
		\frac{\rho_{\ih} \cdot \ADV + \oldd_{\ih}^{\ih} \cdot \frac{\ih}{2 \oldd_{\ih}}}{\ADV+\oldd_{\ih}^{\ih}} \ge \min\left\{\rho_{\ih},\frac{\ih}{2\oldd_{\ih}}\right\} 
			& =\Omega(\log \ih) 
			=\Omega(\log\log h).
	\end{align*}
\end{proof}

The theorem holds also if instead of a single infinite tree $\T$, there is a \emph{finite} tree $\T_k$ for each $k$. The trees only need to be large enough so that, whenever we want to choose an empty subtree, we can instead find a subtree with negligibly small server mass.

\begin{proof}[Proof of \pref{cor:infty}]
Consider the same tree as in the proof of \pref{thm:main}, except that it contains the subtrees $T_i$ for only one value of $i=\ih$. By the identical arguments as in the proof of \pref{thm:main}, we obtain a lower bound of $\Omega(\log \ih)$ for trees of depth $\ih+1$. If the subtrees $T_i$ are included for all $i$, we obtain a metric space with no competitive algorithm for the $\infty$-server problem.
\end{proof}

\section{Conclusions}

We have refuted the $(h,k)$-server hypothesis by proving that,
on trees of sufficient depth, $\cR_\T(h,k) = \Omega(\log \log h)$ 
even when $k$ is arbitrarily large.  When expressed in terms of the depth $D$ of the tree, the lower bound amounts to $\Omega(\log D)$ and applies also to the $\infty$-server
problem.

The construction of our lower bound is inherently fractional: On higher scales, even if an algorithm is deterministic, it can move only a fraction of a group of servers. It would be interesting to show a lower bound for deterministic algorithms that is substantially larger than the randomized one.

Intriguing gaps remain between the lower and upper bounds. The upper bound that would follow from the randomized $k$-server conjecture when disabling the $k-h$ extra servers, $O(\log h)$, is exponentially larger than our lower bound. For deterministic algorithms, the gap is even doubly exponential.

%%%%%%%%%%%%%%%%%%%%%%%%%%%%%%%%%%%%%%%%%%%%%%%%%%%%%%%%%%%%%%%%%%%%%%%%%%%%%%%%%%

\bibliographystyle{alpha}
\bibliography{references}

\end{document}